%% file: main.tex
\documentclass[11pt]{article}
\newif\ifsubmission
\submissiontrue
\submissionfalse %

\newif\ifcomment
\commenttrue
\commentfalse %

\ifsubmission
\commentfalse
\author{Anonymous Authors}
\else
\author{
  Yotam Kenneth-Mordoch
  \\ Weizmann Institute of Science
  \\ \texttt{yotam.kenneth@weizmann.ac.il}
}

\fi

\input{paper.tex}

\title{Faster Pseudo-Deterministic Minimum Cut}
\begin{document}

\maketitle
\begin{abstract}
  Pseudo-deterministic algorithms are randomized algorithms that, with high constant probability, output a fixed canonical solution.
  The study of pseudo-deterministic algorithms for the global minimum cut problem was recently initiated by Agarwala and Varma [ITCS'26],
  who gave a black-box reduction incurring an $O(\log n \log \log n)$ overhead.
  We introduce a natural graph-theoretic tie-breaking mechanism that uniquely selects a canonical minimum cut.
  Using this mechanism, we obtain:
  (i) A pseudo-deterministic minimum cut algorithm for weighted graphs running in $O(m\log^2 n)$ time, eliminating the $O(\log n \log \log n)$ overhead of prior work and matching existing randomized algorithms.
  (ii) The first pseudo-deterministic algorithm for maintaining a canonical minimum cut in a fully-dynamic unweighted graph, with $\mathrm{polylog}(n)$ update time and $\tilde{O}(n)$ query time.
  (iii) Improved pseudo-deterministic algorithms for unweighted graphs in the dynamic streaming and cut-query models of computation, matching the best randomized algorithms.
\end{abstract}
\setcounter{page}{0}
\thispagestyle{empty}
\newpage

\section{Introduction}
\label{sec:introduction}

\input{introduction}

\section{Preliminaries}
\label{sec:preliminaries}
\input{preliminaries}

\section{Pseudo-Deterministic Minimum-Cut Algorithm}
\label{sec:minimum-cut-alg}
\input{minimum-cut-alg}

\section{Additional Models of Computation}
\label{section:additional-models}
\input{additional_models.tex}

\ifsubmission
\else
\subsection*{Acknowledgements}
The author thanks Robert Krauthgamer for helpful comments on earlier drafts of this work.
\fi

{\small
  \bibliographystyle{alphaurl}
  \bibliography{bibliography}
} %

\appendix

\end{document}

%% file: paper.tex
\usepackage[letterpaper, margin=1in]{geometry}
\usepackage{amsthm}
\usepackage{amsmath}
\usepackage{algorithm}
\usepackage[noend]{algpseudocode}

\usepackage{tikz}
\usepackage{amsfonts}
\usepackage{import} 
\usepackage{xifthen}
\usepackage{mathtools}
\usepackage{dsfont}
\usepackage[cal=esstix]{mathalfa}
\usepackage{thm-restate}
\usepackage{makecell} %
\usepackage{booktabs} %
\usepackage{array} %
\newcolumntype{M}[1]{>{\centering\arraybackslash}m{#1}}

\usepackage{comment}
\usepackage[colorinlistoftodos,textsize=tiny,textwidth=2cm,color=red!25!white,obeyFinal]{todonotes}

\usepackage[hypertexnames=false]{hyperref}
\hypersetup{colorlinks=true,allcolors=blue}
\usepackage{cleveref} 

\AddToHook{cmd/appendix/before}{\crefalias{section}{appendix}}

\newcommand{\mintcut}{\mathrm{cut}}
\newcommand{\N}{\mathbb{N}}

\newcommand{\Probability}[1]{\Pr\left[ #1 \right]}

\newcommand{\tO}{\tilde{O}}

\DeclareMathOperator{\poly}{poly}
\DeclareMathOperator{\polylog}{polylog}
\providecommand{\set}[1]{{\{#1\}}}

\ifcomment

\else
\fi

\newtheorem{theorem}{Theorem}[section]

\newtheorem{claim}[theorem]{Claim}

\newtheorem{lemma}[theorem]{Lemma}
\newtheorem{corollary}[theorem]{Corollary}

\newtheorem{definition}[theorem]{Definition}

\newcommand{\LN}[1]{LN(#1)}
\newcommand{\ADD}[2]{#1.ADD(#2)}

\newcommand{\SUBTREE}[2]{#1.SUBTREE(#2)}

\newcommand{\A}{\mathcal{A}}
\newcommand{\priority}{P}

%% file: introduction.tex
The study of pseudo-deterministic algorithms was initiated by \cite{GG11}.
Such algorithms are randomized, yet output a fixed solution with high constant probability.
Formally, we call a randomized algorithm $\A$ \emph{pseudo-deterministic} if for every input $x$, there exists a \emph{canonical output} $s(x)$ such that $\Probability{\A(x)=s(x)} \ge 2/3$.
Such algorithms are often useful when a consistent output is desired, for example when interacting with an adversary or when the algorithm is used as a building block in a larger algorithm.
These attractive properties have motivated the study of pseudo-deterministic algorithms in various contexts, including  matching problems, matroid intersection, undirected connectivity, and more \cite{GG17,OS17,GL19,AV20,GG21,CLORS23}.

The study of pseudo-deterministic algorithms for the global minimum cut problem was only initiated recently by \cite{AV26}.
Their approach is based on the isolation lemma of \cite{MVV87}, and the strategy for its derandomization of \cite{GG17}.
Using this, they show that given a weighted graph on $n$ vertices and $m$ edges, it suffices to make $O(\log n \log \log n)$ calls to any randomized minimum cut algorithm in a black box manner to find a canonical minimum cut with high probability.
Applying this approach to the $O(m\log^2 n)$ time randomized minimum cut algorithm of \cite{GMW20} yields an $O(m\log^3 n \log \log n)$ time algorithm, which is faster than current deterministic algorithms that use $O(m\log^c n)$ time for some large constant $c$ \cite{HLRW24}.%
\footnote{The exact value of $c$ is not explicitly stated in \cite{HLRW24}, but can be inferred to be a large constant.}

In contrast, our approach is based on a graph theoretic tie-breaking mechanism that ensures the uniqueness of the canonical minimum cut.
This approach offers two main advantages.
First, it yields a faster pseudo-deterministic minimum cut algorithm running in $O(m\log^2 n)$ time by modifying the algorithm of \cite{GMW20} to account for the tie-breakers.
Second, our algorithm can also be applied to find the same canonical minimum cut in a contraction of the original graph.%
\footnote{As long as no edge of the canonical minimum cut is contracted.}
We exploit this property to design the first pseudo-deterministic fully-dynamic minimum cut algorithm.
This algorithm is also adversarially-robust, as pseudo-deterministic algorithms are inherently resistant to adaptive updates.

\subsection{Results}
Our main result is a new randomized algorithm for finding a canonical minimum cut in $O(m\log^2 n)$ time.
\begin{theorem}
    \label{theorem:sequential-result}
    There exists a randomized algorithm that returns a canonical minimum cut in a weighted graph $G$ in $O(m\log^2 n)$ time with high probability.
\end{theorem}
Our algorithm improves by a factor of $O(\log n \log \log n)$ over the algorithm of \cite{AV26}, and matches known randomized minimum-cut algorithms \cite{GMW20,MN20}.

Another advantage of our approach is that it uses a tie-breaking mechanism that can be applied consistently to contractions of the original graph.
We exploit this property to design the first pseudo-deterministic fully-dynamic minimum cut algorithm for unweighted graphs, leveraging recent dynamic minimum cut algorithms that are based on graph contractions \cite{GHNSTW23,KK25b,HKMR25}.
\begin{theorem}
    \label{theorem:dynamic-result}
    There exists a randomized fully-dynamic algorithm that, reports a canonical minimum cut in a dynamic unweighted graph $G=(V,E)$ on $n$ vertices with worst-case update time $O(\polylog(n))$ and query time $\tO(n)$.
    The algorithm is randomized and succeeds with high probability.
\end{theorem}
This time bound is optimal up to polylogarithmic factors, as simply outputting the vertices of a minimum cut requires $\Omega(n)$ time.
In comparison, known deterministic minimum cut algorithms take $n^{3/2+o(1)}$ worst-case update time \cite{KK25d,VC25}.
Finally, one of the main advantages of pseudo-deterministic algorithms is their robustness to adaptive updates.
Viewed through this lens, our algorithm also improves upon existing adversarially-robust minimum-cut algorithms, which take $O(n^{1+o(1)})$ worst-case update time \cite{HKMR25}.

The insights of \Cref{theorem:sequential-result} can also be applied to the cut-query and streaming models of computation.
Note that our results in both these models apply only to unweighted graphs, unlike the algorithms of \cite{AV26} that work for weighted graphs as well.
However, our algorithms improve upon the complexity of \cite{AV26} in these settings.
Our first result is an algorithm for finding a canonical minimum cut of an unweighted graph in the dynamic streaming setting using $O(n\log n)$ space and $2$ passes.
\begin{theorem}
    \label{theorem:streaming-result}
    There exists a randomized dynamic streaming algorithm that, given an unweighted graph $G=(V,E)$ on $n$ vertices, finds a canonical minimum cut of $G$ using $O(n\log n)$ bits in $2$ passes with high probability.
\end{theorem}
Our algorithm matches the space and pass complexity of existing randomized minimum cut algorithms in this model \cite{RSW18,ACK19,AD21}.
Compared to the algorithm of \cite{AV26}, our approach requires only $2$ passes instead of $O(\log^2 n)$ passes and achieves an $O(\log^3 n \log \log n)$ factor improvement in space complexity.%
\footnote{The polylogarithmic factors in the space complexity of \cite{AV26} are not explicitly stated but can be inferred to be at least $O(n\log^4 n \log\log n)$ bits from the weighted minimum cut algorithm of \cite{MN20}.}
Finally, there are no known deterministic algorithms for finding a minimum cut in the streaming model using $o(n^2)$ space.

We also present an algorithm for the cut-query model that finds a canonical minimum cut of an unweighted graph using $O(n)$ cut-queries.
\begin{theorem}
    \label{theorem:cut-query-result}
    There exists a randomized algorithm that, given an unweighted graph $G=(V,E)$ on $n$ vertices, finds a canonical minimum cut of $G$ using $O(n)$ cut-queries with high probability.
\end{theorem}
Our algorithm achieves the same query complexity as existing randomized minimum cut algorithms in this model \cite{RSW18,MN20,AEGLMN22}. 
Additionally, the query complexity of our algorithm improves upon \cite{AV26} by a $\polylog n$ factor. 
Finally, our algorithm also improves upon known deterministic minimum cut algorithms in this model, which use $\tilde{O}(n^{5/3})$ cut-queries \cite{ASW25}.
A summary of our results across various computational models appears in Table~\ref{table:results}.

\begin{table}[ht]
    \centering
    \renewcommand{\arraystretch}{2}
    \begin{tabular}{|M{1.8cm}|M{2.1cm}|M{1.8cm}|M{2.4cm}|M{3cm}|M{2.7cm}|}
        \hline
        \textbf{Model} & \textbf{Complexity Measure} & \textbf{Graph Type} & \textbf{Our Pseudo-Det. Results} & \textbf{Randomized} & \textbf{Existing Pseudo-Det.} \cite{AV26} \\
        \hline
        sequential 
        & time
        & weighted 
        & $O(m\log^2 n)$ 
        & \makecell{$O(m\log^2 n)$ \\ \cite{GMW20,MN20}}
        & $O(m\log^3 n)$ \\
        \hline
        fully dynamic 
        & \makecell{worst-case \\ time}
        & unweighted 
        & $\polylog (n)$ update, $\tO(n)$ query
        & \makecell{$\polylog(n)$ update \\ $\tO(n)$ query \\ \cite{HKMR25}}
        & -- \\
        \hline
        dynamic streaming 
        & \makecell{bits, \\ passes}
        & unweighted 
        & \makecell{$O(n\log n)$ \\ $2$ passes}
        & \makecell{$O(n\log n)$ \\ $2$ passes \\ \cite{AD21}}
        & \makecell{$O(n\log^4 n)$ \\ $O(\log^2 n)$ passes} \\
        \hline
        cut query 
        & queries
        & unweighted 
        & $O(n)$ 
        & \makecell{$O(n)$ \\ \cite{AEGLMN22}}
        & $\tO(n)$ \\
        \hline
    \end{tabular}
    \caption{Comparison of our pseudo-deterministic minimum cut algorithms with existing randomized and pseudo-deterministic minimum cut algorithms, across computational models. All complexities omit $O(\log \log n)$ factors.}
    \label{table:results}
\end{table}

\subsection{Technical Overview}
At a high level, our tie-breaking mechanism for returning a specific minimum cut can be described as follows.
It first finds the first vertex (under some ordering) that can be separated from a designated source vertex by a (global) minimum cut, and among all such minimum cuts, it returns the one with the fewest vertices.
To formalize this, begin by designating some vertex $s\in V$ as the \emph{source} and fix an ordering $h:V\setminus\set{s}\to [n-1]$ of the other vertices.
For every cut $S\subseteq V$, define its lexicographic number as $\LN{S}=\min_{x\in S} h(x)$ if $s\not\in S$ and $\LN{S}=\min_{x\in V\setminus S} h(x)$ otherwise.
Similarly, the total vertex priority of a cut $S$, denoted $\priority(S)$, is the sum of priorities of all vertices on the side of the cut that does not contain $s$, i.e.,  $\priority(S)=\sum_{v\in S} \priority(v)$ if $s\not\in S$, and otherwise $\priority(S)=\sum_{v\in V\setminus S} \priority(v)$.
For finding the canonical minimum cut we set $\priority(v)=1$ for all $v\in V$, and hence it is simply the cardinality of the side of the cut that does not contain $s$.%
\footnote{We refer to the cardinality as priority to make terminology consistent throughout the paper, where we later use priorities that are not necessarily $1$.}
We say that a minimum cut $S\subseteq V$ is \emph{lexicographically first} if it minimizes the tuple $(\mintcut_G(S),\LN{S},\priority(S))$ among all minimum cuts with $s\not\in S$, where $\mintcut_G(S)$ is the value of the cut $S$.%
\footnote{A tuple $(a,b,c)$ is considered smaller than a tuple $(a',b',c')$ if $a<a'$, or $a=a'\wedge b<b'$, or $a=a'\wedge b=b'\wedge c<c'$}
Notice that a lexicographically-first cut is always a minimum cut of the graph.
The following claim states that lexicographically-first is a tie-breaker among all minimum cuts.
\begin{claim}
    \label{claim:uniqueness-of-lex-first-cut}
    Let $G=(V,E)$ be a graph with a designated vertex $s\in V$ and an ordering $h:V\setminus\set{s}\to[n-1]$.
    Then, there exists a unique lexicographically-first minimum cut in $G$.
\end{claim}
\begin{proof}
    Let $v\in V\setminus{s}$ be the first vertex, according to $h$, such that there exists some minimum cut separating $v$ from $s$.
    Note that such a vertex always exists since the graph has at least one minimum cut.
    Assume towards contradiction that there exist two distinct minimum cuts $S_1,S_2$ separating $s$ from $v\in S_1,S_2$ and both $S_1,S_2$ have minimum $\priority(S)=|S|$ among all minimum cuts separating $v,s$.
    By submodularity,
    \begin{equation*}
        \mintcut_G(S_1) + \mintcut_G(S_2)
        \ge \mintcut_G(S_1 \cup S_2) + \mintcut_G(S_1 \cap S_2)
        .
    \end{equation*}
    Since $S_1,S_2$ are minimum cuts, we have that both $\mintcut_G(S_1 \cup S_2),\mintcut_G(S_1 \cap S_2) \ge \mintcut_G(S_1)=\mintcut_G(S_2)$, and hence equality must hold, and $S_1\cap S_2$ is a minimum cut separating $v$ from $s$ with $|S_1\cap S_2|<|S_1|=|S_2|$ in contradiction to our assumption.
\end{proof}
We note that finding the minimum cut with minimal total vertex priority for every $v\in V$ is a key component in fast constructions of cactus graphs, succinct representations of all minimum cuts of a graph \cite{KP09,HHS24}.
We observe that one can use these algorithms for cactus graph construction for pseudo-deterministic minimum cut algorithms.
Specifically, using the algorithm of \cite{HHS24} we obtain the following corollary.
\begin{corollary}[\cite{HHS24}]
    There exists a randomized algorithm that given a weighted graph $G=(V,E,w)$ and a vertex $s\in V$, finds for every vertex $v\in V$, the minimum cut separating $v$ from $s$ with minimum total vertex priority (if one exists) in $O(m\log^3 n)$ time with high probability.
\end{corollary}
Combining this corollary with the tie-breaking mechanism above yields a pseudo-deterministic minimum cut algorithm with running time $O(m\log^3 n)$.
Observe that this algorithm already improves the time complexity over \cite{AV26} by an $O(\log \log n)$ factor.

\paragraph{Pseudo-Deterministic Minimum Cut} 
Our main technical contribution (\Cref{theorem:sequential-result}) is a new randomized algorithm that finds the lexicographically-first minimum cut in $O(m\log^2 n)$ time.
The algorithm is actually more general, as it allows for general vertex priorities to be used as a third tie-breaker instead of using $\priority(S)=|S|$.
Given vertex priorities $\priority:V\to\N$, denote the weight of a set of vertices $S$ as $\priority(S)=\sum_{v\in S} \priority(v)$.
Using vertex priorities allows us to use our minimum cut algorithm in the fully-dynamic setting, as we explain in the next section.
Note that setting all vertex priorities to $1$ is equivalent to finding the lexicographically-first minimum cut.
\begin{lemma}
    \label{theorem:modified-karger-algorithm}
    There exists a randomized algorithm that given a graph $G=(V,E)$ a designated source vertex $s$, an ordering $h:V\setminus\set{s}\to[n-1]$ and vertex priorities $\priority:V\to \N$, returns a cut $S\subseteq V\setminus\set{s}$ that minimizes the tuple $(\mintcut_G(S),\LN{S},\priority(S))$.
    The algorithm runs in $O(m\log^2 n)$ time and succeeds with high probability.
\end{lemma}
Our algorithm is based on the $O(m\log^2n)$ time algorithm of \cite{GMW20}.
We modify the algorithm to use our tie-breakers and find the lexicographically-first minimum cut, without increasing the running time of the original algorithm.

\paragraph{Pseudo-deterministic fully-dynamic algorithm.}
Our pseudo-deterministic fully-dynamic minimum cut algorithm (Theorem~\ref{theorem:dynamic-result}) is based on \cite{HKMR25}.
Their main result is a data structure that can be either deterministic or randomized, depending on the known components used in its construction.
It supports edge updates, and upon request constructs a non-trivial minimum cut sparsifier (NMC) of the current graph defined as follows.
\begin{definition}[Non-Trivial Minimum Cut Sparsifier (NMC)]
    \label{definition:nmc}
    A non-trivial minimum cut sparsifier (NMC) of a graph $G=(V,E)$, is a contraction $H=(V_H,E_H)$ of $G$ such that for every $S\subseteq V$ that is a non-trivial minimum cut in $G$, it holds that $\mintcut_H(S)=\mintcut_G(S)$.%
    \footnote{
    We define $\mintcut_H(S)$ as follows.
    For every vertex $u\in V_H$ let $C_u\subseteq V$ be the set of vertices of $G$ that were contracted to form $u$.
    Then, for any cut $S\subseteq V$ in $G$, let $\mintcut_H(S)=\mintcut_G(S)$ if $S$ does not partition any contracted vertex of $H$ and otherwise $\mintcut_H(S)=\infty$.
}
Here, a cut $S\subseteq V$ is called \emph{non-trivial} if $|S|\ne 1,n-1$.
\end{definition}

Our dynamic algorithm first fixes some source vertex $s\in V$ and an ordering $h:V\setminus\set{s}\to[n-1]$ of the other vertices.
It then uses the data structure of \cite{HKMR25} to maintain an NMC of the dynamic input graph $G$, alongside the first vertex in the ordering with minimum degree in $G$.
When queried for a minimum cut, the algorithm constructs the NMC and runs our algorithm of \Cref{theorem:modified-karger-algorithm} on it to find a lexicographically-first minimum cut.
Finally, the algorithm compares the corresponding cut in $G$ to the trivial cut given by the first vertex with minimum degree in $G$, and returns the lexicographically-first among the two.
Notice that since all of our algorithms provide high-probability guarantees, we can use a union bound to show that the algorithm returns the canonical minimum cut with high probability against an adaptive adversary throughout $\poly(n)$ updates and queries.
We argue below that this cut corresponds to the lexicographically-first minimum cut in the original graph.

We note that \cite{HKMR25} also give an adversarially-robust fully-dynamic minimum cut algorithm, but it is slower than ours.
It uses a version of their NMC data structure that is deterministic during edge updates, and leverages randomization only for constructing the NMC.
Since the update procedure is fully deterministic, it is robust to adaptive updates.
Unfortunately, the deterministic components lead update and query time of $n^{o(1)},n^{1+o(1)}$ respectively.
In contrast, using an NMC data structure that is randomized achieves an update time $\polylog(n)$ and query time $\tO(n)$.
This opens the algorithm to adversarial attacks, however we fix this issue using our pseudo-deterministic minimum cut algorithm (applied on the NMC).
The following theorem summarizes the randomized data structure of \cite{HKMR25} and using it we prove Theorem~\ref{theorem:dynamic-result}.
\begin{theorem}[Theorem 1 of \cite{HKMR25}]
    \label{theorem:hkmar-data-structure}
    There is a fully-dynamic algorithm, that given as an input an unweighted graph $G=(V,E)$ on $n$ vertices undergoing edge insertions and deletions, maintains a data structure with $\polylog(n)$ worst-case update time that upon request constructs a non-trivial minimum cut sparsifier (NMC) of $G$ in $\tO(n)$ time with high probability.
    Furthermore, the NMC has $O(n/\delta)$ vertices and $O(n)$ edges, where $\delta$ is the minimum degree of $G$.
\end{theorem}
\begin{proof}[Proof of \Cref{theorem:dynamic-result}]
    Choose deterministically an ordering of the vertices $h:V\to[n]$ and a source vertex $s\in V$, e.g. $h^{-1}(n)$.
    Then, maintain the data structure of \Cref{theorem:hkmar-data-structure} on the dynamic graph $G$, and also maintain a heap storing the degrees of all vertices in $G$, using the ordering $h$ to break ties.
    Notice that by \Cref{theorem:hkmar-data-structure}, the worst-case update time for both data structures is $\polylog(n)$.
    When queried for a minimum cut, construct an NMC $H=(V_H,E_H)$ of $G$ using the data structure of \Cref{theorem:hkmar-data-structure}, and assign each vertex in $H$ a priority equal to the number of vertices of $G$ that were contracted into it, i.e., if $u\in V_H$ corresponds to $C_u\subseteq V$ then set $\priority(u)=|C_u|$.
    In addition, set the lexicographic number of every $u\in V_H$ to be the minimum lexicographic number of a vertex in $C_u$, i.e., $h(u)=\min_{x\in C_u} h(x)$.
    Finally, set the source vertex of $H$ to be the contracted vertex containing $s$.
    To find the lexicographically-first minimum cut of $H$, run \Cref{theorem:modified-karger-algorithm} on it and let $S^*\subseteq V_H$ be the cut returned by the algorithm.
    Furthermore, let $v$ be the vertex with minimum degree in $G$ (using $h$ to break ties).
    Then, if the cut $\set{v}$ is lexicographically smaller than $S^*$ return $\set{v}$ and otherwise return $\bigcup_{u\in S^*} C_u$.

    We now argue that the algorithm returns the lexicographically-first minimum cut in $G$.
    Notice that for every cut $S\subseteq V_H$ we have $(\mintcut_H(S),\LN{S},\priority(S))=(\mintcut_G(S'),\LN{S'},|S'|)$, where $S'=\bigcup_{u\in S} C_u$.
    In particular this holds for every non-trivial minimum cut $T \subseteq V$ since they are all preserved by the definition of an NMC.
    Therefore, if the lexicographically-first cut in $G$ is non-trivial, then it will be found in $H$.
    Otherwise, if the lexicographically-first cut in $G$ is trivial, i.e., a single vertex $v$, then it must be that $\mintcut_G(v)=\delta$ and $\LN{v} \le \LN{S'}$ for every non-trivial minimum cut $S'$ in $G$.
    Therefore, the algorithm returns the lexicographically-first minimum cut in $G$.

    To conclude, we analyze the query time of the algorithm.
    Constructing the NMC takes $\tO(n)$ time by \Cref{theorem:hkmar-data-structure}.
    Setting the vertex priorities and lexicographic numbers in $H$ can be achieved in $O(n)$ time by iterating over the vertices of $G$.
    Finally, the NMC has $O(n)$ edges by \Cref{theorem:hkmar-data-structure}, and hence running \Cref{theorem:modified-karger-algorithm} on $H$ takes $\tO(n)$ time.
    Therefore, the total query time is $\tO(n)$.
\end{proof}

\paragraph{Additional models of computation.} 
Our algorithms for the dynamic streaming and cut-query models leverage our tie-breaking mechanism (\Cref{claim:uniqueness-of-lex-first-cut}) to find the lexicographically-first minimum cut.
We note that both algorithms are limited to unweighted graphs, unlike \cite{AV26} which works for weighted graphs as well. 
However, our approach achieves better complexities. 
The improvements stem from two factors: (i) our tie-breaking mechanism can be implemented using $O(1)$ minimum cut instances instead of $O(\log n)$ repetitions as in \cite{AV26}, and (ii) the black-box reduction of \cite{AV26} invokes weighted minimum cut algorithms, which are more expensive than their unweighted counterparts.

Our algorithms for both models are based on existing randomized algorithms for finding minimum cuts in these models \cite{AD21,AEGLMN22}.
Both algorithms use a weaker version of NMC that preserves each minimum cut with constant probability, rather than preserving all of them with high probability as in the NMC of \Cref{theorem:hkmar-data-structure}.
To obtain a pseudo-deterministic algorithm, construct $k=O(1)$ such sparsifiers independently and assign vertex priorities and lexicographic numbers in the same way as in \Cref{theorem:dynamic-result}.
Then, find in each sparsifier, a lexicographically-first minimum cut using our tie-breaking mechanism. %
\footnote{There may be multiple lexicographically-first cuts since our algorithm assigns general vertex priorities.}
Also, find the lexicographically-first trivial cut, i.e., the vertex with minimum degree in the original graph, using the ordering $h$ to break ties.
Finally, return the lexicographically-first cut among all cuts found.
Notice that if the lexicographically-first cut is non-trivial, it is preserved with constant probability by each sparsifier, and hence at least one of the sparsifiers will preserve it with probability at least $2/3$ when $k$ is sufficiently large.
Otherwise, if the lexicographically-first cut is trivial, it will be found by the algorithm with probability $1$.
The full proof of \Cref{theorem:streaming-result,theorem:cut-query-result} appears in \Cref{section:additional-models}.

%% file: preliminaries.tex
Consider a tree $T=(V,E)$, rooted at a vertex $r\in V$. 
We use the following notation for subtrees.
\begin{definition}
    $v^{\downarrow}$ is the set of vertices that are descendants of $v$ in the rooted tree (including $v$).
\end{definition}

\subsection{Subtree Minimum Query}
In this section we present a data structure that, given a tree and associated vertex values, efficiently returns the minimum value found in a given subtree;
it is based on the standard range minimum query data structure, see e.g. \cite{BF00}.
\begin{definition}[Range Minimum Query (RMQ)]
    Given an array $A$ with $n$ values, a \emph{range minimum query} takes as input two indices $1\le i \le j \le n$ and returns the minimum value in the range $A[i],A[i+1],\ldots,A[j]$;
    \begin{equation*}
        RMQ(i,j)
        =
        \min_{k\in\set{i,i+1,\ldots,j}} A[k]
        .
    \end{equation*}
\end{definition}
\begin{claim}[Folklore]
    There exists a data structure that can be constructed in $O(n)$ time, and answers an RMQ query in $O(1)$ time.
\end{claim}
\begin{corollary}
    \label{corollary:subtree-minimum-query}
    Let $T=(V,E)$ be a spanning tree rooted at some vertex $r\in V$, where each vertex $u\in V$ is associated with some value $a(u)$.
    One can construct in $O(n)$ time a data structure supporting two queries taking $O(1)$ time each:
    (i) Given a vertex $v\in V$ return $\min_{x\in v^{\downarrow}} a(x)$.
    (ii) Given two vertices $v,w\in V$ where $w$ is a descendant of $v$, return $\min_{x\in v^{\downarrow}\setminus w^{\downarrow}} a(x)$.
\end{corollary}
\begin{proof}
    Begin by initializing an empty array $A$ of size $n$ and an index pointer $i=1$.
    Perform a DFS traversal of the tree and whenever encountering a vertex for the first time, insert its value into $A[i]$ and increment $i$.
    Then, construct an RMQ data structure for $A$.
    Furthermore, let $v_1,v_2$ be the values of $i$ when first and last encountering $v$ during the DFS, respectively.
    Observe that $A[v_1], \ldots, A[v_2-1]$ correspond exactly to the values stored in the subtree rooted at $v$ and hence $\min_{x\in v^{\downarrow}} a(x)=RMQ(v_1,v_2-1)$.
    For the second query type, return the minimum of $RMQ(v_1,w_1-1)$ and $RMQ(w_2,v_2-1)$, which correspond to the values in $v^{\downarrow}\setminus w^{\downarrow}$.
    To conclude, notice that performing the DFS and constructing the RMQ data structure takes $O(n)$ time, and each query is answered using $O(1)$ RMQ queries.
\end{proof}

\subsection{Link-Cut Tree}
Our algorithm uses the link-cut tree data structure \cite{ST83}, which is a rooted tree with dynamic edge costs that supports various operations in amortized logarithmic time.
We assume that the costs of the tree are on the vertices, by associating each edge's with its lower vertex; this mapping is bijective since the tree is rooted.
Specifically, we use the variant presented in \cite{GMW20}.
\begin{lemma}[Link-Cut Tree]
    A \emph{link-cut tree} is a data structure that maintains a rooted tree $T$ on a fixed vertex set $V$ with vertex costs, and supports the following operations in amortized $O(\log |V|)$ time.
    \begin{enumerate}
        \item $T.ADD(u,\Delta)$: adds $\Delta$ to the cost of every vertex from $u$ to the root (including $u$).
        \item $T.SUBTREE(u)$: returns the minimum cost vertex (and its cost tuple) in the subtree rooted at $u$.
    \end{enumerate}
\end{lemma}
Note that our algorithms use lexicographically ordered tuple costs $(x_1,\ldots,x_k)$, and that link-cut trees can support tuple costs without increasing time complexity.

%% file: minimum-cut-alg.tex
In this section we prove \Cref{theorem:modified-karger-algorithm}.
We begin by giving an overview of the algorithm of \cite{GMW20} (which follows closely the algorithm of \cite{Karger00}), and then proceed to explain how we modify it to find the lexicographically-first minimum cut.
We say that a cut \emph{$k$-respects} a spanning tree $T$, if the intersection between the edges of $T$ and the those of the cut has cardinality at most $k$.
The algorithm of \cite{GMW20} is composed of two main parts.
First, it finds a packing of $O(\log n)$ trees such that with high probability every minimum cut $2$-respects at least one tree in the packing.
\begin{theorem}[Theorem 3 of \cite{GMW20}]
    \label{theorem:tree-packing-algorithm}
    There exists a randomized algorithm that, given a weighted undirected graph $G$, in $O(m\log^2 n)$ time constructs a set of $O(\log n)$ spanning trees $\mathcal{T}$ such that with high probability every minimum cut $2$-respects at least one tree in $\mathcal{T}$.
\end{theorem}

After obtaining the tree packing, the algorithm uses an efficient dynamic program to find the minimum cut that $2$-respects every tree in the set, and finally it returns the minimum cut found throughout its execution.
Our main technical lemma is a modified algorithm for this second part, which finds a lexicographically-first cut that $2$-respects a given tree, rather than just any minimum cut.
This lemma immediately implies \Cref{theorem:modified-karger-algorithm}.
\begin{lemma}
    \label{lemma:modified-minimum-cut-2-respecting-algorithm}
    There exists an algorithm that, given a spanning tree $T$ of a graph $G$, a source vertex $s$, and an ordering $h:V\setminus\set{s}\to [n-1]$, finds a lexicographically-first minimum cut that $2$-respects $T$ in time $O(m\log n)$.
\end{lemma}
\begin{proof}[Proof of \Cref{theorem:modified-karger-algorithm}]
    Begin by constructing a tree packing $\mathcal{T}$ using \Cref{theorem:tree-packing-algorithm}.
    With high probability the lexicographically-first minimum cut of $G$ $2$-respects some tree in $\mathcal{T}$ by the guarantees of the tree packing algorithm of \cite{GMW20}.
    Therefore, applying \Cref{lemma:modified-minimum-cut-2-respecting-algorithm} on every tree in $\mathcal{T}$ and returning the lexicographically-first minimum cut among all cuts found, yields the lexicographically-first minimum cut of the graph with high probability.
    Finally, the time complexity of the algorithm is $O(m\log^2 n)$ since we run \Cref{lemma:modified-minimum-cut-2-respecting-algorithm} $O(\log n)$ times, and the time complexity of the tree packing is also $O(m\log^2 n)$.
\end{proof}

The rest of the section is devoted to proving \Cref{lemma:modified-minimum-cut-2-respecting-algorithm}.
The algorithm for finding the lexicographically-first cut that $2$-respects the tree is divided into three cases - when it $1$-respects the tree, when it $2$-respects the tree and the edges $(e,e')$ defining the cut are not on a path to the root (independent case), and when it $2$-respects it and the edges $(e,e')$ defining the cut are on some path to the root (descendant case).
It is straightforward to see that every cut that $2$-respects the tree falls into one of these cases.
Note that the proof for the $1$-respecting and independent cases are relatively simple modifications of the algorithms of \cite{GMW20}, while the descendant case requires a more substantial modification.
Finally, throughout this section we root $T$ at the source vertex $s$.

\subsection{$1$-Respecting Cut Case}
For the first case, the algorithm of \cite{GMW20} finds $\mintcut_G(S)$ for every cut $S\subseteq V$ that $1$-respects $T$, and then returns the minimum among them.
Our modified algorithm additionally computes the values $\LN{S},\priority(S)$ for each $1$-respecting cut $S$, and returns the cut minimizing the tuple $(\mintcut_G(S),\LN{S},\priority(S))$.
The following lemma immediately yields the main result of this section.
\begin{lemma}
    \label{lemma:1-respecting-cuts}
    The tuple $(\mintcut_G(S),\LN{S},\priority(S))$ of every cut that $1$-respects a given spanning tree can be computed in time $O(m+n)$.
\end{lemma}
\begin{proof}
    Every cut $S\subseteq V$ that $1$-respects the tree $T$ is of the form $v^{\downarrow}$ for some $v\in V\setminus\set{s}$.
    Therefore, $\priority(S)=\priority(v^{\downarrow})=\sum_{u \in v^{\downarrow}} \priority(u)$, which can be rewritten recursively as $\priority(v^{\downarrow})=\priority(v)+\sum_{u\in children(v)} \priority(u^{\downarrow})$, where $children(v)$ are the children of $v$ in the rooted tree.
    This formula can be efficiently computed for every $v\in V$ using a postorder traversal of the tree in $O(n)$ time.
    To compute the lexicographic number of each $1$-respecting cut, we use a subtree minimum query data structure (\Cref{corollary:subtree-minimum-query}).
    This takes $O(n)$ time overall, as construction takes $O(n)$, plus $O(1)$ time per query.
    Finally, to find the value of each cut we use the following result from \cite{Karger00}.
    \begin{lemma}[Lemma 5.1 of \cite{Karger00}]
        The values of all cuts that $1$-respect a given spanning tree can be determined in $O(m+n)$ time.
    \end{lemma}
    Combining all three values we obtain the tuple $(\mintcut_G(S),\LN{S},\priority(S))$ for every cut $S$ that $1$-respects the tree in time $O(m+n)$.
\end{proof}
\begin{corollary}
    \label{corollary:1-respecting-cuts}
    There exists an algorithm that, given a spanning tree $T$ of a graph $G$, a source vertex $s$, and an ordering $h:V\setminus\set{s}\to [n-1]$, finds a lexicographically-first minimum cut that $1$-respects $T$ in time $O(m + n)$.
\end{corollary}

\subsection{Descendant Case}
We say that a cut is a \emph{minimum descendant cut} of $u$ if it is of the form $u^{\downarrow} \setminus v^{\downarrow}$ for some descendant $v$ of $u$ and has minimum value among all such cuts.
Similarly, a \emph{lexicographically-first descendant cut} is a descendant cut $S\subseteq V$ of $u$  minimizing the tuple $(\mintcut_G(S),\LN{S},\priority(S))$.
Finally, throughout this section we abuse the notation by referring to vertices by their place in the ordering $h$, i.e. given a cut $S\subseteq V$ we denote $\LN{S}=\arg\min_{x\in S} h(x)$.

The main result in this section is an algorithm that finds a lexicographically-first descendant cut for each $u\in V$, extending the approach of \cite{GMW20} with additional tie-breakers.
The algorithm of \cite{GMW20} performs an Euler tour of the tree, and updates a link-cut tree $T'$ on the same edges as $T$ after each step of the tour.
When encountering the vertex $u$ on the way down, the algorithm finds the minimum descendant cut of $u$ by calling $\SUBTREE{T'}{u}$.
Our algorithm employs a similar process, but maintains two link-cut trees $T_1,T_2$ that have the same edges as the spanning tree $T$, each with a different cost tuple.
The first element of both cost tuples is calculated in the same manner as the cost of the link-cut tree of \cite{GMW20}.
The algorithm uses $T_1$ to check whether there exists a minimum descendant cut of $u$ whose lexicographic number is $\LN{u^{\downarrow}}$, that is it achieves the lowest possible lexicographic number of any descendant cut of $u$.
It then updates the edge costs of $T_2$ to refine the search for the lexicographically-first minimum cut, and finds the lexicographically-first minimum descendant cut by calling $\SUBTREE{T_2}{u}$.

We now present the cost tuples of the two link-cut trees $T_1,T_2$.
For every $v\in V\setminus \set{s}$ its cost tuple in $T_1$ is of the form $(a_1,a_2=-\LN{v^{\downarrow}},a_3=\priority(v^{\downarrow}))$, where $a_1$ is the vertex costs as calculated as in \cite{GMW20}.
The cost tuple of $v$ in $T_2$ is of the form $(b_1,b_2,b_3=-\priority(v^{\downarrow}))$ where $b_1$ is again calculated as in \cite{GMW20}, and $b_2$ is initialized to $0$.
The following lemma states that whenever the Euler tour reaches a vertex $v$ on the way down, a call to $\SUBTREE{T'}{v}$ returns a minimum cut of the form $v^{\downarrow} \setminus w^{\downarrow}$ for some descendant $w$ of $v$.
\begin{claim}[Lemma 8 of \cite{GMW20}]
    \label{claim:descendant-subtree-query}
    Let $T'$ be a link-cut tree whose vertex costs are maintained as in \cite{GMW20}.
    When encountering the vertex $u$ for the first time in the Euler tour, a call to $\SUBTREE{T'}{u}$ returns a vertex $v$ such that $u^{\downarrow} \setminus v^{\downarrow}$ is a minimum descendant cut of $u$.
    Furthermore, the exact value of the cut can be computed in $O(1)$ time given $\mintcut_G(u^{\downarrow})$ and $\mintcut_G(v^{\downarrow})$.
    Finally, these costs can be maintained throughout the Euler tour using total time $O(m\log n)$.
\end{claim}
Observe that \Cref{claim:descendant-subtree-query} implies whenever calling $\SUBTREE{T_1}{u}$ or $\SUBTREE{T_2}{u}$ we obtain a minimum descendant cut of $u$ since the first component of both cost tuples is maintained as in \cite{GMW20}.
Therefore, from here on we focus on how the second and third components of the cost tuples of $T_1,T_2$ are used to find the lexicographically-first minimum descendant cut of $u$.

We are now ready to show how our algorithm uses $T_1,T_2$ to find the lexicographically-first descendant cut among all descendant cuts.
It first checks whether there exists a minimum descendant cut that includes the vertex $\LN{u^{\downarrow}}$ using the following claim.
\begin{claim}
    \label{claim:descendant-cut-includes-lex-number}
    Let $(a_1,a_2,a_3)$ be the cost tuple returned by $\SUBTREE{T_1}{u}$.
    Then, $a_2=-\LN{u^{\downarrow}}$ if and only if there exists a minimum descendant cut of $u$ that includes the vertex $\LN{u^{\downarrow}}$.
\end{claim}
\begin{proof}
    Let $W$ be the set of all vertices that form a minimum descendant cut with $u$.
    Then, $\SUBTREE{T_1}{u}=\arg\max_{x\in W} \LN{x^{\downarrow}}$ by \Cref{claim:descendant-subtree-query} (breaking ties according to $a_3$).
    Therefore, $a_2=-\LN{u^{\downarrow}}$ only if every minimum descendant cut of $u$ excludes $\LN{u^{\downarrow}}$ and otherwise $a_2\ne-\LN{u^{\downarrow}}$.
\end{proof}
Next, we explain how the algorithm find the lexicographically-first minimum descendant cut in each of these cases.
For both cases we need the following easy claim.
\begin{claim}
    \label{claim:descendant-scardinality}
    Let $W$ be the set of all descendants of $u$ for which the pair $(b_1,b_2)$ is minimal.
    Then, $\SUBTREE{T_2}{u} = \arg\min_{x \in W} \priority(u^{\downarrow} \setminus x^{\downarrow})$ (breaking ties arbitrarily).
\end{claim}
\begin{proof}[Proof of \Cref{claim:descendant-scardinality}]
    For every descendant $v$ of $u$ the priority of the cut $u^{\downarrow} \setminus v^{\downarrow}$ is given by $\priority(u^{\downarrow}) - \priority(v^{\downarrow})$.
    Since $\priority(u^{\downarrow})$ is fixed, minimizing the priority of the cut $u^{\downarrow} \setminus v^{\downarrow}$ is equivalent to maximizing $\priority(v^{\downarrow})$.
    By our choice of the $b_3$ component of the cost tuple of $T_2$, $b_3=-\priority(v^{\downarrow})$ for every vertex $v\in V$.
    Therefore, among all descendants of $u$ for which $(b_1,b_2)$ is minimal, a call to $\SUBTREE{T_2}{u}$ returns the vertex $v$ with maximum $\priority(v^{\downarrow})$, and accordingly minimum $\priority(u^{\downarrow} \setminus v^{\downarrow})$.
\end{proof}

\paragraph{Case 1: all minimum descendant cuts exclude $\LN{u^{\downarrow}}$.}
Let $v$ be the vertex defining the lexicographically-first minimum descendant cut of $u$.
Observe that $v$ must be on the path between $u$ and $\LN{u^{\downarrow}}$, since $\LN{u^{\downarrow}}\in v^{\downarrow}$.
Let $x$ be the lowest vertex on the path from $u$ to $\LN{u^{\downarrow}}$ that forms a minimum descendant cut with $u$.
Observe that $x$ must be lower than $v$ along the path (or $x=v$), and hence $u^{\downarrow} \setminus v^{\downarrow} \subseteq u^{\downarrow} \setminus x^{\downarrow}$.
Therefore, $\LN{u^{\downarrow} \setminus v^{\downarrow}} \ge \LN{u^{\downarrow} \setminus x^{\downarrow}}$ and equality is achieved by the choice of $v$.
Given $x$ and $\LN{u^{\downarrow} \setminus x^{\downarrow}}$, one can find the lexicographically-first minimum cut as follows.
First, call $\ADD{T_2}{x,(0,-1,0)}$ and $\ADD{T_2}{\LN{u^{\downarrow} \setminus x^{\downarrow}},(0,1,0)}$ on the second tree.
This sets the $b_2$ component of the cost tuple of every vertex on the path from $x$ to $\LN{u^{\downarrow} \setminus x^{\downarrow}}$ to $-1$, and leaves the $b_2$ component of every other vertex as $0$.
Therefore, calling $\SUBTREE{T_2}{u}$ returns a minimum descendant cuts on the path descending from $\LN{u^{\downarrow} \setminus x^{\downarrow}}$ to $x$.
Notice that such a cut exists since $v$ is a vertex on the path descending from $\LN{u^{\downarrow} \setminus x^{\downarrow}}$ to $x$ (including $x$) by our choice of $x,v$.
Then, calling $\SUBTREE{T_2}{u}$ returns the vertex representing the cut with minimum $\priority(S)$ among all minimum descendant cuts on the path by \Cref{claim:descendant-scardinality}.
Therefore, the vertex returned is $v$ by our assumption that $v$ forms a lexicographically-first minimum descendant cut of $u$.
Finally, the algorithm calls $\ADD{T_2}{x,(0,1,0)}$ and $\ADD{T_2}{\LN{u^{\downarrow} \setminus x^{\downarrow}},(0,-1,0)}$ to reset the state of $T_2$.

To conclude, we explain how to find $x$.
Observe that $\SUBTREE{T_1}{u}$ returns a vertex on the path from $u$ to $\LN{u^{\downarrow}}$ by our assumption that all minimum descendant cuts of $u$ exclude $\LN{u^{\downarrow}}$.
Furthermore, for every vertex $w$ on the path from $u$ to $\LN{u^{\downarrow}}$ we have $\LN{w^{\downarrow}}=\LN{u^{\downarrow}}$ by the minimality of $\LN{u^{\downarrow}}$.
Therefore, the tie-breaker is by the $a_3$ component of the cost tuple, which equals $\priority(w^{\downarrow})$ for every descendant $w$ of $u$, and is monotonically decreasing along the path from $u$ to $\LN{u^{\downarrow}}$.
Hence, since $\SUBTREE{T_1}{u}$ minimizes the tuple $(a_1,a_2,a_3)$ it returns the lowest vertex along the path from $u$ to $\LN{u^{\downarrow}}$ that forms a minimum descendant cut of $u$, which is exactly the vertex $x$ defined above.
Then, given $x$ we can find $\LN{u^{\downarrow} \setminus x^{\downarrow}}$ using \Cref{corollary:subtree-minimum-query} in $O(1)$ time.

\paragraph{Case 2: there exists a minimum descendant cut that includes $\LN{u^{\downarrow}}$.}
The algorithm calls $\ADD{T_2}{\LN{u^{\downarrow}},(0,1,0)}$ on $T_2$, returns the cut $u^{\downarrow} \setminus \SUBTREE{T_2}{u}^{\downarrow}$, and finally calls $\ADD{T_2}{\LN{u^{\downarrow}},(0,-1,0)}$ to reset the state of $T_2$.
Let $v$ be the vertex returned by $\SUBTREE{T_2}{u}$ and denote its corresponding cut by $S=u^{\downarrow} \setminus v^{\downarrow}$.
We now show that $S$ is a lexicographically-first minimum descendant cut of $u$.
First, observe that $S$ is a minimum descendant cut by \Cref{claim:descendant-subtree-query}.
The first tie-breaker of $\SUBTREE{T_2}{u}$ is the component $b_2$, which equals $1$ along the path from $\LN{u^{\downarrow}}$ to the root, and $0$ elsewhere.
Observe that this enforces that the search space of $\SUBTREE{T_2}{u}$ is exactly all descendants $x$ of $u$ such that $u^{\downarrow} \setminus x^{\downarrow}$ is a minimum descendant cut of $u$ and $x$ is not on the path from $\LN{u^{\downarrow}}$ to the root.
Therefore, $\LN{u^{\downarrow}}\not\in v^{\downarrow}$ and we find $\LN{u^{\downarrow} \setminus v^{\downarrow}}=\LN{u^{\downarrow}}$.
In addition, by \Cref{claim:descendant-scardinality} the cut returned minimizes $\priority(S)$ among all cuts considered.
Finally, since $\LN{u^{\downarrow}}$ is the smallest possible lexicographic number of any descendant cut of $u$, the cut returned is lexicographically-first.
The algorithm for finding the lexicographically-first descendant cut for each vertex $u\in V$ is summarized in \Cref{alg:descendant-cut}.

\begin{algorithm}[ht]
\caption{Lexicographically-First Descendant Cut}
\label{alg:descendant-cut}
\begin{algorithmic}[1]
\Procedure{FindDescendantCuts}{$T, G, h$}
    \State Initialize link-cut trees $T_1, T_2$ on the edges of $T$
    \State Set cost tuples: $T_1$ with $(a_1, -\LN{v^{\downarrow}}, \priority(v^{\downarrow}))$, $T_2$ with $(b_1, 0, -\priority(v^{\downarrow}))$ for each $v$
    \For{each vertex $u$ encountered on the way down during an Euler tour of $T$}
        \State Update $a_1, b_1$ components as in \cite{GMW20}
        \State $(a_1, a_2, a_3) \gets \SUBTREE{T_1}{u}$
        \If{$a_2= -\LN{u^{\downarrow}}$}\Comment{\textbf{Case 1}: all min descendant cuts exclude $\LN{u^{\downarrow}}$}
            \State $x \gets \SUBTREE{T_1}{u}$ \Comment{Lowest vertex on path to $\LN{u^{\downarrow}}$ with min cut}
            \State $\ell \gets \LN{u^{\downarrow} \setminus x^{\downarrow}}$ \Comment{Using \Cref{corollary:subtree-minimum-query}}
            \State $\ADD{T_2}{x, (0, -1, 0)}$; $\ADD{T_2}{\ell, (0, 1, 0)}$ \Comment{Restrict to path from $\ell$ to $x$}
            \State $v \gets \SUBTREE{T_2}{u}$
            \State $\ADD{T_2}{x, (0, 1, 0)}$; $\ADD{T_2}{\ell, (0, -1, 0)}$ \Comment{Reset $T_2$}
        \Else \Comment{\textbf{Case 2}: some min descendant cut includes $\LN{u^{\downarrow}}$}
            \State $\ADD{T_2}{\LN{u^{\downarrow}}, (0, 1, 0)}$ \Comment{Exclude descendants of $\LN{u^{\downarrow}}$}
            \State $v \gets \SUBTREE{T_2}{u}$
            \State $\ADD{T_2}{\LN{u^{\downarrow}}, (0, -1, 0)}$ \Comment{Reset $T_2$}
        \EndIf
        \State Record the cut $u^{\downarrow} \setminus v^{\downarrow}$ with tuple $(\mintcut_G(u^{\downarrow} \setminus v^{\downarrow}), \LN{u^{\downarrow} \setminus v^{\downarrow}}, \priority(u^{\downarrow} \setminus v^{\downarrow}))$
    \EndFor
    \State \Return the cut minimizing the recorded tuple across all $u$
\EndProcedure
\end{algorithmic}
\end{algorithm}

The following lemma states that the algorithm indeed finds the lexicographically-first minimum cut among all descendant cuts, by considering all pairs found in case 1 and case 2.
\begin{lemma}
    \label{lemma:descendant-cuts}
    Let $(u_1,v_1),\ldots,(u_n,v_n)$ be the pairs found by the above process, and let $S_i = u_i^{\downarrow} \setminus v_i^{\downarrow}$ be their corresponding cuts.
    Then, the cut minimizing the tuple $(\mintcut_G(S_i),\LN{S_i},\priority(S_i))$ among all $i\in[n]$ is a lexicographically-first minimum descendant cut.
    Furthermore, the algorithm runs in time $O(m\log n)$.
\end{lemma}
\begin{proof}
    Notice that for each $u\in V$ the argument above shows that by either case 1 or case 2, the algorithm returns a cut that is lexicographically-first among all cuts of the form $u^{\downarrow} \setminus v^{\downarrow}$ for some descendant $v$ of $u$.
    Therefore, the lexicographically-first minimum descendant cut is among the cuts returned by the algorithm.

    It remains to bound the running time of the algorithm.
    The algorithm performs an Euler tour of the tree, which takes $O(n)$ time.
    In addition, it maintains two link-cut trees $T_1,T_2$.
    Note that the $a_1,b_1$ components of the cost tuples can be maintained in $O(m \log n)$ time using \Cref{claim:descendant-subtree-query}.
    In addition, the $a_2,a_3$ components of the first tree and the $b_3$ component of the second are static throughout the execution and can be set using \Cref{lemma:1-respecting-cuts} in $O(n)$ time.

    We now analyze the time spent during case 1 and case 2.
    In case 1 we perform four ADD operations and one SUBTREE operation on $T_2$ each requiring $O(\log n)$ amortized time, and in addition $O(1)$ calls to an RMQ data structure.
    Hence, the overall time spent in case 1 over all vertices is $O(n\log n)$.
    Similarly, in case 2 we perform two ADD operations and one SUBTREE operation on $T_2$ each requiring $O(\log n)$ amortized time.
    Therefore, the overall time spent in case 2 over all vertices is $O(n\log n)$.
    Finally, computing the tuples $(\mintcut_G(S_i),\LN{S_i},\priority(S_i))$ of each cut can be done in $O(1)$ time using \Cref{claim:descendant-subtree-query} and \Cref{corollary:subtree-minimum-query}.
    To conclude, the overall running time of the algorithm is $O(m\log n)$.
\end{proof}

\subsection{Independent Case}
In this section, we show how to find a lexicographically-first minimum cut defined by two independent edges in the tree.
The main result of this section is the following lemma.
\begin{lemma}
    \label{lemma:independent-case}
    There exists an algorithm that, given a graph $G=(V,E)$ with $n$ vertices and $m$ edges, a spanning tree $T$ of $G$, a source vertex $s$, and an ordering $h:V\setminus\set{s}\to [n-1]$, finds the pair of edges $(e,e')$ defining an independent cut $S\subseteq V$ minimizing the tuple $(\mintcut_G(S),\LN{S},\priority(S))$.
    The algorithm runs in time $O(m\log n + n\log n)$.
\end{lemma}
The proof of the lemma uses the bipartite problem framework of \cite{GMW20}.
\begin{definition}[The bipartite problem of \cite{GMW20}]
    Given two trees $T_1,T_2$ with edge costs and a list of non-tree edges $L=\set{(u,v): u\in T_1,v\in T_2}$ along with their costs, find a pair of edges $e\in T_1,e'\in T_2$ that minimize the sum of costs of $e,e'$ plus the sum of costs of edges in $L$ crossing the cut defined by $e,e'$.
    The size of the problem is defined as $|T_1|+|T_2|+|L|$.
\end{definition}
\begin{lemma}[Lemma 11 of \cite{GMW20}]
    \label{lemma:bipartite-problem-solution}
    There exists an algorithm that given an instance of the bipartite problem of size $N$ finds the optimal solution in time $O(N\log N)$.
\end{lemma}
\begin{lemma}[Lemma 10 of \cite{GMW20}]
    There exists an algorithm that, given an edge-weighted graph $G$ with $n$ vertices and $m$ edges, a spanning tree $T$ on the vertices of $G$, reduces finding a minimum independent cut for every tree edge $e$ to multiple instances of the bipartite problem, with total size $O(m)$ in $O(m\log n)$ time.%
    \footnote{The original lemma in \cite{GMW20} returns only a single minimum cost pair $(e,e')$, however the algorithm finds the minimizing pair for every edge $e$ in $T_1$. 
    We find the minimizing pair for every edge $e$ in $T_2$ by reversing the roles of $T_1,T_2$ and calling the algorithm again.}
\end{lemma}
We modify the reduction to use $\priority(S)$ as a tie-breaker.
This is achieved by changing the costs of edges in the trees $T_1,T_2$ into tuples of the form $(c,\priority(v^{\downarrow}))$ where $c$ is the original cost of the edge and $v$ is the lower vertex of the edge.
Notice that the priority of the cut defined by $(e,e')$ is exactly $\priority(v^{\downarrow}) + \priority(v'^{\downarrow})$ where $v,v'$ are the lower vertices of $e,e'$ respectively.
Therefore, the second component of the cost captures the total priority of the cut defined by the edges in the trees.
In addition, we modify the cost of edges outside the trees to $(c,0)$ since they do not define the priority of the cut.

In order to calculate the modified weights efficiently, we use the fact that the trees $T_1,T_2$ are composed of edges of the original tree $T$.
Therefore, there are at most $n-1$ of them across all the trees and their second component can be computed in  $O(n+m)$ time using \Cref{lemma:1-respecting-cuts}.
We are now ready to prove \Cref{lemma:independent-case}.
\begin{proof}[Proof of \Cref{lemma:independent-case}]
    Begin by running the reduction of \cite{GMW20} with the modification detailed above.
    Notice that it is possible to modify the costs of all edges in each modified tree problem by first updating the costs of the original tree edges, and then applying the reduction on the tree with modified costs.
    The reduction still takes $O(m\log n)$ time as the modified edge costs can be computed in $O(n+m)$ time using \Cref{lemma:1-respecting-cuts}.
    In addition, the total size of all bipartite problems remains $O(m)$ since we only modified the edge costs, and therefore they can be solved in $O(m\log n)$ time by \Cref{lemma:bipartite-problem-solution}.

    We now show how to compute the tuple $(\mintcut_G(S),\LN{S},\priority(S))$ for every cut $S\subseteq V$ that was returned by the algorithm.
    Observe that an independent cut defined by edges $(e,e')$ with lower vertices $u,v$ respectively, can be written as $S = u^{\downarrow} \cup v^{\downarrow}$.
    Therefore, the lexicographic number of the cut is given by $\min(\LN{u^{\downarrow}},\LN{v^{\downarrow}})$ and $\priority(S)=\priority(u^{\downarrow}) + \priority(v^{\downarrow})$.
    Using the values calculated in \Cref{lemma:1-respecting-cuts}, we can compute these values for a cut in $O(1)$ time.
    Finally, there are at most $O(n)$ such cuts as the algorithm returns one candidate per edge in the tree, and hence one can compute the cost tuples of all cuts and return the lexicographically-first of them in $O(n)$ time.

    It remains to show that the cut returned is lexicographically first among all independent cuts.
    Let $(u^*,v^*)$ define a lexicographically-first independent cut $S^* = u^{*\downarrow} \cup v^{*\downarrow}$.
    Assume without loss of generality that $\LN{u^{*\downarrow}} < \LN{v^{*\downarrow}}$.
    Examine the cut $T$ found when the algorithm considers the tree edge $e^*$ whose lower vertex is $u^*$.
    Notice that $\mintcut_G(T) \le \mintcut_G(S^*)$ since the algorithm finds the minimum independent containing $e^*$, and therefore equality is obtained since $S^*$ is a minimum independent cut.
    Furthermore, $\LN{T} \le \LN{u^{*\downarrow}} = \LN{S^*}$ as $u^{*\downarrow}\in T$, and again equality is achieved since $S^*$ is lexicographically first.
    Finally, $\priority(T) \le \priority(S^*)$ since it is used as a tie-breaker in the algorithm.
    Therefore, $(\mintcut_G(S^*),\LN{S^*},\priority(S^*))=(\mintcut_G(T),\LN{T},\priority(T))$ and the algorithm returns a lexicographically-first independent cut.
\end{proof}

\subsection{Putting It All Together}
In this section we prove \Cref{lemma:modified-minimum-cut-2-respecting-algorithm} by combining the three cases detailed above.
\begin{proof}[Proof of \Cref{lemma:modified-minimum-cut-2-respecting-algorithm}]
    Throughout, we assume that the input graph is connected, as otherwise we can find the connected component containing each vertex in $O(n)$ time using DFS or BFS.
    Then, we return the connected component with minimum lexicographic number that does not contain $s$.
    Notice that if $m<n-1$ then the graph is disconnected and hence we may assume that $m\ge O(n)$.

    Assuming the above, given a tree $T$, we find the lexicographically-first minimum cut of each of the three cases, using \Cref{corollary:1-respecting-cuts,lemma:independent-case,lemma:descendant-cuts}.
    Notice that every minimum cut that $2$-respects the tree is covered by one of the three cases.
    Therefore, the minimum among all cuts found is indeed the lexicographically-first minimum cut that $2$-respects the tree.
    Finally, the overall running time of the algorithm is $O(m\log n)$ since each of the three cases runs in this time.
\end{proof}

%% file: additional_models.tex
In this section we prove \Cref{theorem:streaming-result,theorem:cut-query-result}.
Both proofs rely on a weaker version of the non-trivial minimum cut sparsifier (NMC) defined in \Cref{definition:nmc}.
Instead of preserving all minimum cuts, this weaker version only preserves each cut with some constant probability.
We note that in general this weaker version is also called a non-trivial minimum cut sparsifier, but we denote it as weak to differentiate between the two versions.
\begin{definition}[Weak Non-Trivial Minimum Cut Sparsifier]
    \label{definition:weak-nmc}
    A \emph{weak non-trivial minimum cut sparsifier} (weak NMC) of a graph $G=(V,E)$ is a randomized contraction $H=(V_H,E_H)$ of $G$ such that for every non-trivial minimum cut $S\subseteq V$ in $G$, it holds that $\mintcut_H(S)=\mintcut_G(S)$ with constant probability.
\end{definition}
The following lemma shows that by constructing $O(1)$ independent weak NMCs and applying our tie-breaking mechanism to each, we can find the lexicographically-first minimum cut with probability at least $2/3$.

\begin{lemma}[Meta-Lemma]
    \label{lemma:weak-nmc-meta}
    Let $G=(V,E)$ be an unweighted graph with a designated source $s\in V$ and an ordering $h:V\setminus\set{s}\to[n-1]$.
    Given $O(1)$ independent weak NMCs $H_1,\ldots,H_k$ of $G$, each with $\kappa$ edges, and the vertex with minimum degree (breaking ties by $h$), there exists an algorithm that finds the lexicographically-first minimum cut of $G$ can be found with probability at least $2/3$.
\end{lemma}
\begin{proof}
    Begin by setting the vertex order and priorities of the NMC sparsifiers, we do this similarly to the proof of \Cref{theorem:dynamic-result}.
    Fix some NMC sparsifier $H_i$, then for each contracted vertex $u\in V_{H_i}$ corresponding to $C_u\subseteq V$, set $\priority(u)=|C_u|$ and $h(u)=\min_{x\in C_u} h(x)$.
    Then, run our lexicographically-first minimum cut algorithm (\Cref{theorem:modified-karger-algorithm}) on each $H_i$ to find its lexicographically-first cut $S_i^*$.
    Finally, compare all candidate cuts and return the lexicographically-first among them.

    We now show that the algorithm finds the lexicographically-first minimum cut with probability at least $2/3$.
    Let $S^*$ be the lexicographically-first minimum cut of $G$.
    If $S^*$ is trivial (i.e., $|S^*|=1$), then it is preserved explicitly as the vertex with minimum degree (breaking ties by $h$).
    If $S^*$ is non-trivial, then by \Cref{definition:weak-nmc}, each weak NMC preserves $S^*$ with constant probability $p>0$.
    With $k=O(1)$ independent weak NMCs, the probability that at least one preserves $S^*$ is at least $1-(1-p)^k \ge 4/5$ for sufficiently large $k$.
    When $S^*$ is preserved by some $H_i$, our algorithm returns $S^*$ with high probability.
    Therefore, the overall probability of returning $S^*$ is at least $2/3$.
\end{proof}

\subsection{Proof of \texorpdfstring{\Cref{theorem:streaming-result}}{Theorem 1.3}}
\label{subsection:streaming-proof}
For the proof of \Cref{theorem:streaming-result}, we need the following result from \cite{AD21} for constructing weak NMCs in the dynamic streaming model.
\begin{theorem}[\cite{AD21}]
    \label{theorem:streaming-weak-nmc}
    There exists a dynamic streaming algorithm that, given an unweighted graph $G=(V,E)$ on $n$ vertices presented as a dynamic stream, constructs a weak NMC $H=(V_H,E_H)$ of $G$ with $O(n)$ edges using $O(n\log n)$ bits in two passes, with high probability.
\end{theorem}
\begin{proof}[Proof of \Cref{theorem:streaming-result}]
    Fix an ordering $h:V\to[n]$ of the vertices deterministically and choose an arbitrary source vertex $s\in V$, e.g. set $s=h^{-1}(n)$.
    Our algorithm construct $k=O(1)$ independent weak NMCs using \Cref{theorem:streaming-weak-nmc}.
    It also tracks the vertex with minimum degree (using $h$ to break ties).
    Using these two inputs we can apply \Cref{lemma:weak-nmc-meta} to find the lexicographically-first minimum cut.

    To conclude, we analyze the complexity of the algorithm.
    Each weak NMC is constructed using $O(n\log n)$ bits in $2$ passes.
    In addition, tracking the vertex with minimum degree requires $O(n\log n)$ bits (storing the degree and $h$-value of the current minimum of every vertex).
    Therefore, the overall space is $O(n\log n)$ bits and the algorithm uses $2$ passes.
\end{proof}

\subsection{Proof of \texorpdfstring{\Cref{theorem:cut-query-result}}{Theorem 1.4}}
\label{subsection:cut-query-proof}
The proof of \Cref{theorem:cut-query-result} follows similarly to that of \Cref{theorem:streaming-result}.
We rely on the minimum cut algorithm of \cite{AEGLMN22} to construct a weak NMC in the cut-query model, however we need to modify their construction to return a weak NMC instead of a minimum cut in certain instances.
We first prove \Cref{theorem:cut-query-result} using the following modified version of their result, and then prove the modified version.
\begin{theorem}[\cite{AEGLMN22}]
    \label{theorem:cut-query-weak-nmc}
    There exists a cut-query algorithm that, given an unweighted graph $G=(V,E)$ on $n$ vertices, constructs a weak NMC $H=(V_H,E_H)$ of $G$ with $O(n)$ edges using $O(n)$ cut-queries, with high probability.
\end{theorem}
\begin{proof}[Proof of \Cref{theorem:cut-query-result}]
    Fix an ordering $h:V\to[n]$ of the vertices deterministically and choose an arbitrary source vertex $s\in V$, e.g. set $s=h^{-1}(n)$.
    The algorithm begins by constructing $k=O(1)$ independent weak NMCs using \Cref{theorem:cut-query-weak-nmc}.
    In addition, it queries the degree of every vertex to find the vertex with minimum degree (using $h$ to break ties).
    Using these two inputs it applies \Cref{lemma:weak-nmc-meta} to find the lexicographically-first minimum cut with probability at least $2/3$.

    To conclude, we analyze the query complexity of the algorithm.
    Construction of the NMCs requires $O(n)$ cut-queries each, and finding the vertex with minimum degree requires $O(n)$ cut-queries (one per vertex).
    Therefore, the overall query complexity is $O(n)$ cut-queries.
\end{proof}
The proof of \Cref{theorem:cut-query-weak-nmc} relies on the algorithm of \cite{AEGLMN22}.
We also need the following standard results for the cut-query model.
\begin{lemma}[\cite{RSW18,PRW24}]
    \label{lemma:cut-query-sparsifier}
    There exists a cut-query algorithm that, given a (possibly weighted) graph $G=(V,E)$ on $n$ vertices, constructs a cut sparsifier $H=(V_H,E_H)$ of $G$ with $O(n\log n/\epsilon^2)$ edges using $O(n\log^2 n/\epsilon^2)$ cut-queries, with high probability.
\end{lemma}
\begin{lemma}[Corollary 2.3 of \cite{KK25c}, based on \cite{BM11}]
    \label{lemma:cut-query-reconstruction}
    There exists a deterministic algorithm that, given cut-query access to a weighted graph $G$ on $n$ vertices and $m$ edges, returns the entire graph $G$ using $O(n+m)$ cut-queries.
\end{lemma}
\begin{proof}[Proof of \Cref{theorem:cut-query-weak-nmc}]
    The algorithm of \cite{AEGLMN22} is partitioned into three cases, depending on the minimum degree $\delta_G$ of the input graph $G$.
    The cases are: (i) $\delta_G \le 5\cdot 10^6$, (ii) $5\cdot 10^6$ and $\log^10 n$, and (iii)  $\delta_G \ge \log^10 n$.
    In the first two cases, the algorithm returns an explicit weak NMC sparsifier of the graph and then finds a minimum cut using any offline algorithm.
    In the third case, the algorithm finds a contraction of the vertex set $V$ into $O(n/\delta_G)$ super-vertices, and then returns a minimum cut by calling the algorithm of \cite{MN20} on the contracted graph.
    Note that in the cut-query model it is possible to simulate cut-queries on a contracted graph using cut queries on the original graph.
    
    We modify the algorithm of the third case to return a weak NMC sparsifier instead of a minimum cut.
    Denote the contracted graph by $G'=(V',E')$ and observe that $G'$ has $O(n/\log^9 n)$ vertices.
    To construct the NMC sparsifier, we will use the following algorithm of \cite{RSW18}.
    \begin{enumerate}
        \item Let $H$ be a simple graph on $N$ vertices.
        \item Compute a quality $(1\pm 1/100)$-cut sparsifier $H'$ of $H$.
        \item Find all non-singleton cuts of size at most $(1+3/100)$ times the minimum cut in $H'$, and contract every edge that does not participate in any such cut.
    \end{enumerate}
    It is shown in \cite{RSW18} that applying this algorithm on a simple graph $H$ with $N$ vertices results in an NMC sparsifier of $H$ with high probability, and that it has at most $O(N)$ edges.
    Therefore, recovering the new contracted graph explicitly using \Cref{lemma:cut-query-reconstruction} yields a weak NMC sparsifier of $G$ with at most $O(N)$ edges.
    Let $G''$ be the output of applying this algorithm on $G'$.
    Notice that $G''$ is the same as the output of applying the algorithm of \cite{AEGLMN22} on $G'$, up to the contraction of small that had already occured in $G$.
    However, every minimum cut that was preserved in $G'$ is also preserved in $G''$, and hence $G''$ is an NMC sparsifier of $G'$.
    Furthermore, it is clear that the number of edges in $G''$ is at most $O(n)$.

    To conclude, we analyze the query complexity of the algorithm.
    Denoting the number of vertices in $G'$ by $N$, constructing the cut sparsifier $H'$ requires $O(N\log^2 N)=O(n/\log^9 n \cdot \log^2 (n/\log^9 n))=O(n/\log^6 n) < O(n)$ cut-queries by \Cref{lemma:cut-query-sparsifier}.
    Then, since the new smaller contracted graph has at most $O(n)$ edges, we can recover it explicitly using $O(n)$ cut-queries by \Cref{lemma:cut-query-reconstruction}.
    Therefore, the overall query complexity of the modified algorithm is $O(n)$ cut-queries.
\end{proof}